\documentclass[sn-mathphys,Numbered]{sn-jnl}% Math and Physical Sciences Reference Style
\usepackage{graphicx}%
\usepackage{multirow}%
\usepackage{amsmath,amssymb,amsfonts}%
\usepackage{amsthm}%
\usepackage{mathrsfs}%
\usepackage[title]{appendix}%
\usepackage{xcolor}%
\usepackage{textcomp}%
\usepackage{manyfoot}%
\usepackage{booktabs}%
\usepackage{algorithm}%
\usepackage{algorithmicx}%
\usepackage{algpseudocode}%
\usepackage{listings}%
\theoremstyle{thmstyleone}%
\newtheorem{theorem}{Theorem}%  meant for continuous numbers
\theoremstyle{thmstyletwo}%
\newtheorem{example}{Example}%
\newtheorem{remark}{Remark}%

\theoremstyle{thmstylethree}%
\newtheorem{definition}{Definition}%

\begin{document}

\title[Article Title]{Explicit Constructions of Optimal $(r,\delta)$-Locally Repairable Codes}


\author[1,2]{\fnm{Yaxin} \sur{Wang}}\email{lhyx123698@163.com}

\author*[1,2]{\fnm{Siman} \sur{Yang}}\email{smyang@math.ecnu.edu.cn}


\affil*[1]{\orgdiv{School of Mathematical Sciences}, \orgname{East China Normal University}, \orgaddress{\street{No. 500, Dong Chuan Road}, \city{Shanghai}, \postcode{200241},\country{China}}}

\affil[2]{\orgdiv{Shanghai Key Laboratory of PMMP}, \orgname{East China Normal University}, \orgaddress{\street{No. 500, Dong Chuan Road}, \city{Shanghai}, \postcode{200241},  \country{China}}}

\abstract{Locally repairable codes (LRCs) have recently been widely used in distributed storage systems and the LRCs with $(r,\delta)$-locality ($(r,\delta)$-LRCs) attracted a lot of interest for tolerating multiple erasures. Ge \emph{et al.} constructed $(r,\delta)$-LRCs with unbounded code length and optimal minimum distance when $\delta+1 \leq d \leq 2\delta$ from the parity-check matrix equipped with the Vandermonde structure, but the block length is limited by the size of $\mathbb{F}_q$. In this paper, we propose a more general construction of $(r,\delta)$-LRCs through the parity-check matrix. Furthermore, with the help of MDS codes, we give three classes of explicit constructions of optimal $(r,\delta)$-LRCs with block length beyond $q$. It turns out that 1) our general construction extends the results of Ge \emph{et al.} and 2) our explicit constructions yield some optimal $(r,\delta)$-LRCs with new parameters.}




\maketitle

\section{Introduction}\label{sec1}
In this era of big data, distributed storage systems have emerged to meet the needs of large-scale storage applications. In order to ensure the reliability of stored data, locally repairable codes(LRCs) are widely used in distributed storage systems.

In 2012, Gopalan \emph{et al.} \cite{bib1} first proposed the concept of LRC. An $[n,k]$ linear code $\mathcal{C}$ is called an LRC with locality $r$ if every code symbol of $\mathcal{C}$ can be repaired by accessing at most $r$ other symbols.  The singleton-type bound for the minimum distance of $[n,k,d]$-LRC with locality $r$ was also given in \cite{bib1} as follows:
\begin{equation}\label{01}
d\leq n-k-\lceil\frac{k}{r}\rceil+2.
\end{equation}
The codes whose minimum distance meets this bound are called optimal. Many constructions of optimal LRCs
have been reported in \cite{bib2,bib3,bib4,bib5,bib6}. However, the locality $r$ is only suitable for repairing a single erasure. Prakash \emph{et al.}\cite{bib7} introduced the concept of $(r,\delta)$-locality for repairing multiple erasures.

For a linear code $\mathcal{C}$ with code length $n$, the $i$-th code symbol of $\mathcal{C}$ is said to have $(r,\delta)$-locality ($\delta\geq 2$) if there exists a subset $S_i\subseteq [n]$ such that: 1) $i\in S_i$, 2) $|S_i|\leq r+\delta-1$ and 3) the punctured code $\mathcal{C}|_{S_i}$ has the minimum distance $d(\mathcal{C}|_{S_i})\geq \delta$. The code $\mathcal{C}$ is said to have $(r,\delta)$-locality or be an $(r,\delta)$-LRC if every code symbol has $(r,\delta)$-locality. Noted that the $(r,\delta=2)$-LRC degenerates to an LRC with locality $r$. It was also proved in \cite{bib7} that an $(r,\delta)$-LRC with parameters $[n,k,d]$ must obey the following Singleton-type bound:
\begin{equation}\label{02}
d\leq n-k+1-(\lceil\frac{k}{r}\rceil-1)(\delta-1).
\end{equation}
In this paper, an $(r,\delta)$-LRC meeting the above bound is called optimal. Constructions of optimal $(r,\delta)$-LRCs attract a lot of interest and many constructions through the parity-check matrix approach have been obtained. 

Guruswami \emph{et al.}\cite{bib4} equipped the parity-check matrix with the Vandermonde structure and constructed optimal $(r,\delta=2)$-LRCs with unbounded code length and minimum distance $d=3,4$. Furthermore, they proved that the code length of a $q$-ary $(r,\delta=2)$-LRCs with minimum distance $d\geq 5$ is upper bounded by $O(q^3)$ and present the construction of $(r,\delta=2)$-LRCs with super-linear  lengths by a greedy algorithm. With the help of extremal graph theory, Xing \emph{et al.}\cite{bib19} improve the known
results in \cite{bib4} for $d\geq7$ through the parity-check matrix with the Vandermonde structure. Later in \cite{bib11}, Ge \emph{et al.} generalized the work of Xing to construct optimal $(r,\delta\geq2)$-LRCs by the parity-check matrix of the following form
\begin{equation}\label{03}
H=\begin{pmatrix}
	U_1&O&\cdots&O\\
	O&U_2&\cdots&O\\
	\vdots&\vdots&\ddots&\vdots\\
	O&O&\cdots&U_m\\
	V_1&V_2&\cdots&V_m
\end{pmatrix},
\end{equation}
where $U_i$ and $\begin{pmatrix}
U_i\\V_i
\end{pmatrix}$ are both Vandermonde matrices. When $\delta+1\leq d\leq 2\delta$, they constructed optimal $(r,\delta)$-LRCs with unbounded code length. A similar construction also appeared in \cite{bib16}. The parity-check matrix $H$ is divided into $m$ disjoint column blocks corresponding to $m$ local repair groups. In this paper, the length of every column block is called block length. Although the code length in \cite{bib11} and \cite{bib16} is unbounded which seems to be ideal, the block length is limited by $q$ due to the Vandermonde structure. Recently, \cite{bib12} presented a construction of optimal $(r,\delta=3)$-LRC with unbounded code length and the block length reaches $q+1$. Motivated by the work, this paper explores further constructions of optimal $(r,\delta)$-LRCs with block length beyond the limitation of $q$. We give a general construction of optimal $(r,\delta)$-LRC with minimum distance $\delta\leq d\leq 2\delta$ which extends the construction in \cite{bib11}. Furthermore, we also construct three classes of explicit optimal $(r,\delta)$-LRCs whose block length can reach $q+1$ or $q+2$.  

The rest of this paper is organized as follows. In Section \ref{sec2}, we provide
some preliminaries on MDS codes. In Section \ref{sec3}, A general construction of optimal $(r,\delta)$-LRCs is given. In Section \ref{sec4}, we present three classes of explicit constructions of optimal $(r,\delta)$-LRCs by using MDS codes.  Finally, Section \ref{sec6} concludes the paper.

\section{Preliminaries}\label{sec2}
In this paper, we apply MDS codes whose code length exceeds $q$ to produce explicit optimal $(r,\delta)$-LRCs. We first introduce some notions and known results about MDS codes in this section.

\begin{definition}[\cite{bib13}]\label{defn1}
	 A linear code with parameters $[n,k,d]$ such that $k+d=n+1$ is called a maximum distance separable (MDS) code.
\end{definition}

For convenience, we define the MDS property as follows.

\begin{definition}
	Let $m\leq n$ be any positive integers. A matrix $A\in \mathbb{F}_q^{m\times n}$ is said to have MDS property if any $m$ column vectors of $A$ are linearly independent. In other words, $A$  is a parity-check matrix of an $[n,k,d]$-MDS code such that $m=n-k=d-1$.
\end{definition}

\begin{example}[\cite{bib14}]\label{ex1}
	Let $q>2$ be a prime power. Let $\alpha$ be a primitive element of $\mathbb{F}_q$ and $f(x)$ be a polynomial on $\mathbb{F}_q$. Define a matrix in $\mathbb{F}_q^{2\times (q+1)}$ by
	\begin{equation}\label{04}
		G_1=\begin{pmatrix}
			f(\alpha)&f(\alpha^2)&\cdots&f(\alpha^{q-1})&1&0\\
			\alpha&\alpha^2&\cdots&\alpha^{q-1}&0&1
		\end{pmatrix},
	\end{equation}
	where $f(x)$ satisfies the two conditions below:\\
	1)\ $f(x)\ne 0$ for all $x\in \mathbb{F}_q^{*}$, and\\
	2)\ $yf(x)-xf(y)\ne 0$ for all distinct $x$ and $y\in\mathbb{F}_q^{*}$.\\
	Then the code $\mathcal{C}$ with the parity-check matrix $G_1$ is a $[q+1,q-1,3]$-MDS code.
\end{example} 

 \begin{example}[\cite{bib15}]\label{ex2}
	Let $q>3$ be a power of 2. Let $\alpha$ be a primitive element of $\mathbb{F}_q$. Define a matrix in $\mathbb{F}_q^{3\times(q+2)}$ by
	\begin{equation}\label{05}
		G_2=\begin{pmatrix}
			1&1&\cdots&1&1&0&0\\
			\alpha&\alpha^2&\cdots&\alpha^{q-1}&0&1&0\\
			\alpha^2&(\alpha^2)^2&\cdots&(\alpha^{q-1})^2&0&0&1
		\end{pmatrix},
	\end{equation}
	then the code $\mathcal{C}$ with the parity-check matrix $G_2$ is a $[q+2,q-1,4]$-MDS code. 
\end{example}

\section{General Construction}\label{sec3}

In this section, we propose a more general construction of optimal $(r,\delta)$-LRC. We still make use of the parity-check matrix as the form in \eqref{03}, but the matrices $U_i$ and $V_i$ are no longer required to satisfy the property as strong as the Vandermonde structure.

\begin{theorem}\label{thm1}
	Let $r,\delta,d,m$ be positive integers such that $\delta\leq d\leq 2\delta$ and  $r>d-\delta$. Define the following matrix on $\mathbb{F}_q$
	\begin{equation}\label{06}
		H=\begin{pmatrix}
			U_1&O&\cdots&O\\
			O&U_2&\cdots&O\\
			\vdots&\vdots&\ddots&\vdots\\
			O&O&\cdots&U_m\\
			V_1&V_2&\cdots&V_m
		\end{pmatrix},
	\end{equation}
	where $U_i$ is a $(\delta-1)\times (r+\delta-1)$ matrix and $V_i$ is a $(d-\delta)\times (r+\delta-1)$ matrix for any $1\leq i\leq m$. If $U_i$ and $\begin{pmatrix}
		U_i\\V_i
	\end{pmatrix}$ both have MDS property for any $1\leq i\leq m$, then the code $\mathcal{C}$ with $H$ as the parity-check matrix is an optimal $(r,\delta)$-LRC with parameters $[n=m(r+\delta-1), k=rm-(d-\delta) ,d]$.
% Furthermore, $\mathcal{C}$ can be explicitly constructed as long as $U_i$ and $V_i$ are	explicitly given.
\end{theorem}

\begin{proof}
	The code length is clear. Since $U_i$ has MDS property, any $\delta-1$ column vectors of $U_i$ are linearly independent. This implies that the code with $U_i$ as the parity-check matrix has the minimum distance at least $\delta$. It is obvious that $\mathcal{C}$ has $(r,\delta)$-locality.
	
	Suppose that there exist row vectors $\boldsymbol{\lambda}_i\in \mathbb{F}_q^{\delta-1},1\leq i\leq m$ and $\boldsymbol{\lambda}_{m+1}\in \mathbb{F}_q^{d-\delta}$ such that
$$	\boldsymbol{\lambda}_1(U_1,O,\cdots,O)+\cdots+\boldsymbol{\lambda}_m(O,O,\cdots,U_m)+\boldsymbol{\lambda}_{m+1}(V_1,V_2,\cdots,V_m)=\mathbf{0},$$

		then we have
		$$\boldsymbol{\lambda}_iU_i+\boldsymbol{\lambda}_{m+1}V_i=(\boldsymbol{\lambda}_i,\boldsymbol{\lambda}_{m+1})\begin{pmatrix}
			U_i\\V_i
		\end{pmatrix}=\mathbf{0}, \text{ for any}\ 1\leq i\leq m.$$
		By assumption,	$\begin{pmatrix}
			U_i\\ V_i
		\end{pmatrix}$ is row full rank. This implies that 
		$$\boldsymbol{\lambda}_{m+1}=\mathbf{0}\  \text{and}\ \boldsymbol{\lambda}_i=\mathbf{0} \text{ for any}\ 1\leq i\leq m.$$
		Therefore, $H$ is row full rank and we obtain the dimension of $\mathcal{C}$ as follows:
		$$n-rank(H)=n-(m(\delta-1)+d-\delta)=rm-(d-\delta).$$
		
		By the singleton-type bound \eqref{02}, we have
		\begin{align*}
			d(\mathcal{C})&\leq n-k+1-(\lceil\frac{k}{r}\rceil-1)(\delta-1)\\
			%			&\leq n-rm+(d-\delta)+1-(\lceil m-\frac{d-\delta}{r}\rceil-1)(\delta-1)\\
			&\leq(\delta-1)m+(d-\delta)+1-(m-1)(\delta-1)\\
			%			&=\delta-1+d-\delta+1\\
			&=d,
		\end{align*}
		when $r>d-\delta$.
		It suffices to prove the minimum distance is at least $d$, i.e., any $d-1$ columns of $H$ are linearly independent. Select arbitrary $d-1$ columns from $H$.
		
		It is clear that any $d-1$ columns in the same block are linearly independent as $\begin{pmatrix}
			U_i\\V_i
		\end{pmatrix}$ has MDS property.
		
		If the $d-1$ columns are distributed in different blocks, we assume that these blocks are $B_{i_1},B_{i_2},\cdots,B_{i_s}$, $i_1<i_2<\cdots<i_s$. Since $d\leq 2\delta$, there is at most one block containing at least $\delta$ columns. Consider the following two cases separately. 
		
		Case 1: Every block contains at most $\delta-1$ columns. Consider the following matrix consisting of the former $m(\delta-1)$ components of the selected $d-1$ columns :
		$$\begin{pmatrix}
			D_{i_1}&O&\cdots&O\\
			O&D_{i_2}&\cdots&O\\
			\vdots&\vdots&\ddots&\vdots\\
			O&O&\cdots&D_{i_s}
		\end{pmatrix},$$
		where $D_{i_j}$ consists of some columns of $U_{i_j}$ for $1\leq j\leq s$. Since $U_{i_j}$ has MDS property which implies any $\delta-1$ columns of $U_{i_j}$ are linearly independent, $D_{i_j}$ is column full rank. Combining this with the clear fact that the columns of $D_{i_j}, 1\leq j\leq s$ are mutually linearly independent, the above matrix is column full rank. Thus, the selected $d-1$ columns are linearly independent.
		
		Case 2: There is one block containing at least $\delta$ columns. Without loss of generality, we assume that the block is $B_{i_1}$. Consider the following matrix consisting of the selected $d-1$ columns:
		$$\begin{pmatrix}
			D_{i_1}&O&\cdots&O\\
			O&D_{i_2}&\cdots&O\\
			\vdots&\vdots&\ddots&\vdots\\
			O&O&\cdots&D_{i_s}\\
			E_{i_1}&E_{i_2}&\cdots&E_{i_s}
		\end{pmatrix},$$
		where $D_{i_j}$ and $V_{i_j}$ consist of some columns of $U_{i_j}$ and $E_{i_j}$ for $1\leq j\leq s$, respectively. Suppose that there exist column vectors $\boldsymbol{\lambda}_j,1\leq j\leq s$ such that
		$$\begin{pmatrix}
			D_{i_1}\\
			O\\
			\vdots\\
			O\\
			E_{i_1}
		\end{pmatrix}\boldsymbol{\lambda}_1+\begin{pmatrix}
			O\\
			D_{i_2}\\
			\vdots\\
			O\\
			E_{i_2}
		\end{pmatrix}\boldsymbol{\lambda}_2+\cdots+\begin{pmatrix}
			O\\
			O\\
			\vdots\\
			D_{i_s}\\
			E_{i_s}
		\end{pmatrix}\boldsymbol{\lambda}_s=\mathbf{0},$$

		then we have the system of equations:
		$$\begin{cases}
			D_{i_j}\boldsymbol{\lambda}_j=\mathbf{0}\quad \text{for} \ 1\leq j\leq s,\\
			\sum_{j=1}^{s}E_{i_j}\boldsymbol{\lambda}_j=\mathbf{0}.
		\end{cases}$$
		According to Case 1, $D_{i_j}$ is column full rank for $2\leq j\leq s$, which implies $\boldsymbol{\lambda}_i=\mathbf{0}$ when $2\leq j \leq s$. We substitute them into the above equations to obtain
		$$\begin{pmatrix}
			D_{i_1}\\E_{i_1}
		\end{pmatrix}\boldsymbol{\lambda}_1=\mathbf{0}.$$
		Since $\begin{pmatrix}
			D_{i_1}\\E_{i_1}
		\end{pmatrix}$ consists of some columns of $\begin{pmatrix}
			U_{i_1}\\V_{i_1}
		\end{pmatrix}$ and any $d-1$ columns of the matrix $\begin{pmatrix}
			U_{i_1}\\V_{i_1}
		\end{pmatrix}$ are linearly independent due to the MDS property,  we have $\boldsymbol{\lambda}_1=\mathbf{0}$. Thus, the selected $d-1$ columns are linearly independent.
	\end{proof}
	
	\begin{remark}
		Let $D$ be an $m\times n$ Vandermonde matrix with $m\leq n$. Any $m$ columns of $D$ are linearly independent. Thus, $D$ satisfies MDS property. Clearly, the construction of \cite{bib11} and \cite{bib16} is a special case of Theorem \ref{thm1}.
	\end{remark}

\section{Explicit Constructions}\label{sec4}
Based on the characterization of the parity-check matrix as shown in Theorem \ref{thm1}, one can construct an optimal $(r,\delta)$-LRC as long as there exist matrices $\begin{pmatrix}
	U_i\\V_i
\end{pmatrix}$ and $U_i$ with MDS property. 

Combining Theorem \ref{thm1} with Example \ref{ex1}, we have the following result. The proof is immediate from Theorem \ref{thm1} and is omitted here.

\begin{theorem}\label{thm2}
	Let $r>1,m$ be any positive integer and $q$ is a prime power such that $q\geq r+1$. Let $\alpha$ be a primitive element of $\mathbb{F}_q$. Define the following matrix:
	\begin{equation}
		H_1=\begin{pmatrix}
			U_1&O&\cdots&O\\
			O&U_1&\cdots&O\\ 			
			\vdots&\vdots&\ddots&\vdots\\
			O&O&\cdots&U_1
		\end{pmatrix},
	\end{equation}
	where 
	\begin{equation}
		U_1=\begin{pmatrix}
			f(\alpha)&f(\alpha^2)&\cdots&f(\alpha^r)&1&0\\
			\alpha&\alpha^2&\cdots&\alpha^r&0&1
		\end{pmatrix}
	\end{equation}
	and $f(x)$ is a polynomial on $\mathbb{F}_q$ such that
	
	1)\ $f(x)\ne 0$ for all $x\in \mathbb{F}_q^{*}$, and
	
	2)\ $yf(x)-xf(y)\ne 0$ for all distinct $x$ and $y\in\mathbb{F}_q^{*}$.\\
	Let $H_1$ contains $m$ blocks $U_1$. If $\mathcal{C}$ is a linear code with the parity-check
	matrix $H_1$, then $\mathcal{C}$ is an optimal $(r,\delta=3)$-LRC with parameters $[n=(r+2)m,k=rm,d=3]$. Especially when $r=q-1$, the block length reaches $q+1$.
\end{theorem}

\begin{remark}
 Obviously, the constant polynomial $f(x)\equiv 1$ satisfies the two conditions mentioned in Theorem \ref{thm2} and the corresponding construction is just the construction in \cite[Theorem 1]{bib12}.
\end{remark}

In order to illustrate the above construction better, an example is given below.

\begin{example}
	Let $q=5$, $r=q-1=4$ and $m=2$. Let $\alpha=2$ which is a primitive element of $\mathbb{F}_5$. Choose the constant polynomial $f(x)\equiv 1$ and construct the following matrix:
	$$\begin{pmatrix}
		\begin{array}{cccccc|cccccc}
			1&1&1&1&1&0  &0&0&0&0&0&0\\
			2&4&3&1&0&1  &0&0&0&0&0&0\\
			0&0&0&0&0&0  &1&1&1&1&1&0\\
			0&0&0&0&0&0  &2&4&3&1&0&1
		\end{array}
	\end{pmatrix},$$
	then the code $\mathcal{C}$ with the above matrix as the parity-check matrix is an $(r=4,\delta=3)$-LRC with parameters $[12,8,3]$. By the Singleton-type bound
	$$n-k+1-(\lceil\frac{k}{r}\rceil-1)(\delta-1)=12-8+1-(2-1)(3-1)=3,$$
	the minimum distance is optimal.
\end{example}

Combining Theorem \ref{thm1} with Example \ref{ex2}, we have the following result.

\begin{theorem}
	Let $r>1,m$ be any positive integer and $q$ is a power of 2 such that $q\geq max\{4,r+1\}$. Let $\alpha$ be a primitive element of $\mathbb{F}_q$. Define the following matrix:
	\begin{equation}
		H_2=\begin{pmatrix}
			U_2&O&\cdots&O\\
			O&U_2&\cdots&O\\
			\vdots&\vdots&\ddots&\vdots\\
			O&O&\cdots&U_2
		\end{pmatrix},
	\end{equation}
	where 
	\begin{equation}
		U_2=\begin{pmatrix}
			1&1&\cdots&1&1&0&0\\
			\alpha&\alpha^2&\cdots&\alpha^{r}&0&1&0\\
			\alpha^{2}&(\alpha^2)^2&\cdots&(\alpha^{r})^2&0&0&1
		\end{pmatrix}.
	\end{equation}
	Let $H_2$ contains $m$ blocks $U_2$. If $\mathcal{C}$ is a linear code with the parity-check
	matrix $H_2$, then $\mathcal{C}$ is an optimal $(r,\delta=4)$-LRC with parameters $[n=(r+3)m,k=rm,d=4]$. Especially when $r=q-1$, the block length reaches $q+2$.
\end{theorem}

\begin{example}
	Let $q=2^2$, $r=q-1=3$ and $m=2$. Let $\alpha$ be a primitive element of $\mathbb{F}_4$ which is the root of the irreducible polynomial $x^2+x+1$ on $\mathbb{F}_2$. Construct the following matrix:
	$$\begin{pmatrix}
		\begin{array}{cccccc|cccccc}
			1&1&1&1&0&0 &0&0&0&0&0&0\\
			\alpha&\alpha^2&1&0&1&0&0&0&0&0&0&0\\
			\alpha^2&\alpha&1&0&0&1&0&0&0&0&0&0\\
			0&0&0&0&0&0&1&1&1&1&0&0\\
			0&0&0&0&0&0&\alpha&\alpha^2&1&0&1&0\\
			0&0&0&0&0&0&\alpha^2&\alpha&1&0&0&1
		\end{array}
	\end{pmatrix},$$
	then the code $\mathcal{C}$ with the above matrix as the parity-check matrix is an $(r=3,\delta=4)$-LRC with parameters $[12,6,4]$. By the Singleton-type bound
	$$n-k+1-(\lceil\frac{k}{r}\rceil-1)(\delta-1)=12-6+1-(2-1)(4-1)=4,$$
	the minimum distance is optimal.	
\end{example}

Next, we give the other construction by making use of the matrix $G_2$ in  \eqref{05}. By removing the last column from $G_2$, we can get the submatrix 
$$G_3=\begin{pmatrix}
		1&1&\cdots&1&1&0\\
		\alpha&\alpha^2&\cdots&\alpha^{q-1}&0&1\\
		\alpha^2&(\alpha^2)^2&\cdots&(\alpha^{q-1})^2&0&0
	\end{pmatrix}.$$
Obviously, $G_3$ also has MDS property. Consider its submatrix $G_4$ consisting of the first two rows of $G_3$, it is easy to verify that any two columns of $G_4$ are linearly independent. Thus, $G_4$ has MDS property. Next we apply $G_3$ and $G_4$ to Theorem \ref{thm1} to obtain the following result.

\begin{theorem}
	Let $r>1,m$ be any positive integer and $q$ be a power of $2$ such that $q\geq max\{4,r+1\}$. Let $\alpha$ be a primitive element of $\mathbb{F}_q$. Define the following matrix:
	\begin{equation}
		H_3=\begin{pmatrix}
			U_3&O&\cdots&O\\
			O&U_3&\cdots&O\\
			\vdots&\vdots&\ddots&\vdots\\
			O&O&\cdots&U_3\\
			V_3&V_3&\cdots&V_3
		\end{pmatrix},
	\end{equation}
	where 
	\begin{equation}
		U_3=\begin{pmatrix}
			1&1&\cdots&1&1&0\\
			\alpha&\alpha^2&\cdots&\alpha^{r}&0&1\\
		\end{pmatrix}, 
	\end{equation}
	\begin{equation}
		V_3=\begin{pmatrix}
			\alpha^{2}&(\alpha^2)^2&\cdots&(\alpha^{r})^2&0&0\\
		\end{pmatrix}.
	\end{equation}
	Let $H_3$ contains $m$ blocks $U_3$. If $\mathcal{C}$ is a linear code with the parity-check
	matrix $H_3$, then $\mathcal{C}$ is an optimal $(r,\delta=3)$-LRC with parameters $[n=(r+2)m,k=rm-1,d=4]$. Especially when $r=q-1$, the block length reaches $q+1$.
\end{theorem}

\begin{example}
	Let $q=2^2$, $r=q-1=3$ and $m=2$. Let $\alpha$ be a primitive element of $\mathbb{F}_4$ which is the root of the irreducible polynomial $x^2+x+1$ on $\mathbb{F}_2$. Construct the following matrix:
	$$\begin{pmatrix}
		\begin{array}{ccccc|ccccc}
			1&1&1&1&0   &0&0&0&0&0   \\
			\alpha&\alpha^2&1&0&1   &0&0&0&0&0  \\
			0&0&0&0&0   &1&1&1&1&0    \\
			0&0&0&0&0   &\alpha&\alpha^2&1&0&1    \\
			\hline
			\alpha^2&\alpha&1&0&0   &\alpha^2&\alpha&1&0&0  
		\end{array}
	\end{pmatrix},$$
	then the code $\mathcal{C}$ with the above matrix as the parity-check matrix is an $(r=3,\delta=3)$-LRC with parameters $[10,5,4]$. Since
	$$n-k+1-(\lceil\frac{k}{r}\rceil-1)(\delta-1)=10-5+1-(2-1)(3-1)=4,$$
	the minimum distance is optimal.	
\end{example}

\section{Conslusions}\label{sec6}
In this paper, through the parity-check matrix with MDS property, we present a general construction of optimal $(r,\delta)$-LRC. Compared with the construction  in \cite{bib11} and \cite{bib16}, our construction requires a weeker condition on the parity-check matrix. Furthermore, Three classes of explicit constructions of optimal $(r,\delta)$-LRCs are given by applying known MDS codes. When $r=q-1$, their block length can reach $q+1$ or $q+2$ which is longer than those constructions by employing the Vandermonde structure.

\bibliography{sn-bibliography}

\end{document}